\documentclass{article}

\usepackage{arxiv}

\usepackage[utf8]{inputenc} 
\usepackage[T1]{fontenc}    
\usepackage{hyperref}       
\usepackage{url}            
\usepackage{booktabs}       
\usepackage{amsfonts}       
\usepackage{nicefrac}       
\usepackage{microtype}      
\usepackage{lipsum}

\usepackage{amsthm}
\usepackage{graphicx}
\usepackage{enumerate}

\newtheorem{lemma}{Lemma}[section]

\theoremstyle{definition}
\newtheorem{definition}{Definition}[section]

\title{Interval-permutation segment graphs}

\author{
  \bf{Zlatko~Joveski} \\
  \texttt{zlatko.joveski@vanderbilt.edu} \\ \\ 
  \qquad
  \bf{Jeremy P.~Spinrad} \\
  \texttt{jeremy.p.spinrad@vanderbilt.edu} \\ \\
  Department of Electrical Engineering and Computer Science\\
  Vanderbilt University\\
  Nashville, TN 37235\\
}

\begin{document}
\maketitle

\begin{abstract}
In this work we introduce the \textit{interval permutation segment (IP-SEG)} model that naturally generalizes the geometric intersection models of interval and permutation graphs. We study properties of two graph classes that arise from the IP-SEG model and present a family of forbidden subgraphs for these classes. In addition, we present polynomial algorithms for the clique and independent set problems on these classes, when the model is given as part of the input.  
\end{abstract}



\section{Introduction}
\label{sec:intro}

Many important graph classes are defined or can be characterized by a geometric intersection model. These characterizations can arise from different geometric objects that are being intersected. Examples include boxicity graphs (intersection graphs of $d$-dimensional rectangles), circular-arc graphs (intersection graphs of arcs along a circle), and string graphs (intersection graphs of curves in the plane) \cite{Perfect, Efficient}. Two particularly well-studied examples are the classes of interval and permutation graphs \cite{IntervalFishburn}. In both of their respective models, the intersecting objects are line segments in the plane, with different restrictions imposed on their positions. In interval graphs, the line segments must lie on a single line, while in permutation graphs, their endpoints must lie on two separate parallel lines. Because of the similarity, it is natural to look for geometric intersection models that would generalize those of interval and permutation graphs. One approach is to have geometric objects that generalize line segments. In the model of \textit{simple triangle} graphs (also known in the literature as \textit{point-interval}) \cite{SimpleTriangle, PointInterval}, the intersecting objects are triangles, while in the model of \textit{trapezoid} graphs (also referred to as \textit{interval-interval}) \cite{Trapezoid, PointInterval}, the intersecting objects are trapezoids. 

Another way of generalizing the models of interval and permutation graphs is to use the same kind of intersecting geometric objects - straight line segments, but reduce the restrictions on their possible positions. If we drop all restrictions, then we obtain the large class of SEG graphs - the intersection graphs of straight line segments in the plane \cite{Survey}. However, many of the standard optimization problems, including independent set \cite{IndClqIntKratochvil}, remain NP-hard on SEG graphs. Sub-classes of SEG graphs, with restrictions on the number of directions that line segments could have, have been studied, including grid intersection (or 2-DIR) graphs \cite{GridHartman}. Such models, however, do not generalize that of permutation graphs since segments there can have any direction in the plane, except being parallel with the two lines. 

Here we introduce a new model in which the intersecting objects remain straight line segments, but they can either lie along one of two horizontal lines or go from one horizontal line to the other. This leads to two natural generalizations of both interval and permutation graphs, based on whether all horizontal segments lie on the same line. We formally define the models and graph classes in Section \ref{sec:prels}. In Section \ref{sec:rep_and_count}, we show that these classes have implicit representations. Unlike simple triangle and trapezoid graphs, the two new classes are not contained in the class of perfect graphs. However, in Section \ref{sec:cycles}, we show that we are somewhat limited in how we can represent chordless cycles using the new model, which helps us identify some forbidden subgraphs for the graph classes. In Section \ref{sec:algorithms}, we present polynomial algorithms for the clique and independent set problems, when the model is known. In Section \ref{sec:conclusion}, we discuss some of the open questions on the new graph classes.

\section{Preliminaries}
\label{sec:prels}

Let $G = (V,E)$ be a graph with a vertex set $V = \{ v_1, v_2, \dots, v_n \}$. An intersection model of $G$ is a family of sets $S_i,\ i = 1,2,\dots,n,$ such that $(v_i, v_j) \in E$ if and only if $S_i \cap S_j \neq \emptyset$. We say that $G$ is an intersection graph of the family of sets $S_i$.

Many classes of graphs are defined or can be characterized as the intersection graphs of different types of families of sets. For example, line graphs are the intersection graphs of edges of graphs and circular arc graphs are the intersection graphs of arcs on a circle. In this work, we are primarily interested in interval and permutation graphs and their generalizations. Figure \ref{fig:int_perm_geom_ex} shows the respective geometric intersection models of a graph that is both interval and permutation.

\begin{definition}
    Interval graphs are the intersection graphs of intervals on the real line.
\end{definition}

\begin{definition}
    Permutation graphs are the intersection graphs of line segments whose endpoints lie on two parallel lines $L_1$ and $L_2$, so that for each segment $s$, its endpoints $s_1$ and $s_2$ lie on $L_1$ and $L_2$, respectively.
\end{definition}

\begin{figure}[h]
\centering
\includegraphics[scale=.95]{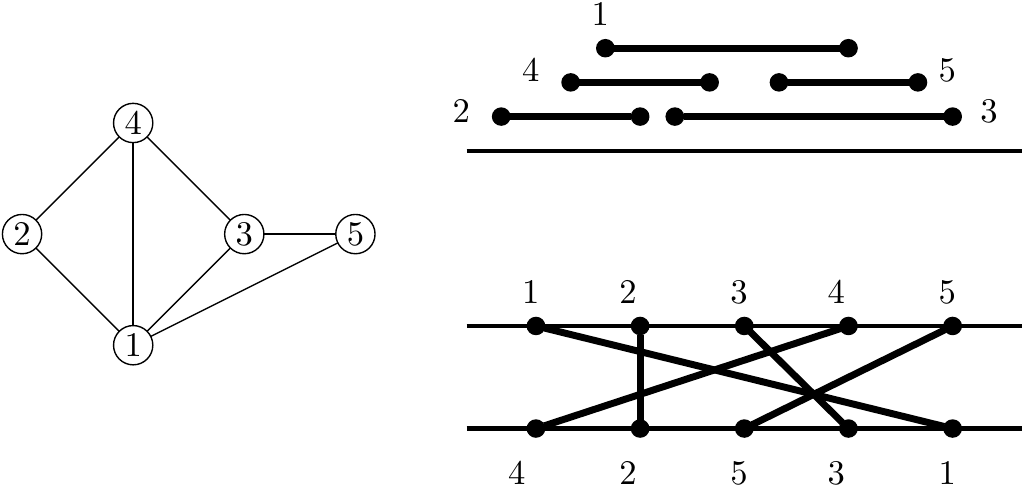}
\caption{A graph that is both permutation and interval, along with the respective geometric representations.}
\label{fig:int_perm_geom_ex}
\end{figure}

Note that we can obtain an equivalent definition of interval graphs by dropping the requirement that the line in question is \textit{real} and by substituting \textit{intervals} with \textit{line segments}. This allows for multiple natural ways of simultaneously generalizing the geometric intersection models of interval and permutation graphs. These include \textit{point-interval} (or \textit{simple triangle}) and \textit{interval-interval} (or \textit{trapezoid}) graphs introduced by Corneil and Kamula \cite{PointInterval}, in which the geometric objects corresponding to vertices are formed by combining multiple line segments. In this work, we study a different generalization of permutation and interval graphs in which the geometric objects corresponding to vertices remain individual line segments.

\begin{definition}
    Let $L_1$ and $L_2$ and be two parallel lines. We say that a line segment $s$ is an \textit{interval segment} if both its endpoints $s_1$ and $s_2$ lie on the same line $L_i$. We say that $s$ is a \textit{permutation segment} if $s_1$ lies on $L_1$ and $s_2$ lies on $L_2$. Let $I$ be a set of interval segments and $P$ a set of permutation segments, and let $G$ be the intersection graph of the set of segments $I \cup P$. We call $I \cup P$ an \textit{interval-permutation-segment} (\textit{IP-SEG}) model of $G$. If all interval segments in $I$ lie on the same parallel line $L_i$, we call $I \cup P$ an \textit{IP-SEG*} model of $G$.
\end{definition}

\begin{definition}
    A graph $G$ is an \textit{IP-SEG} \textit{graph} if it has an IP-SEG model. $G$ is an \textit{IP-SEG*} \textit{graph} if it has an IP-SEG* model.
\end{definition}

Note that under the above definitions, an IP-SEG* model is also an IP-SEG model, which indicates that IP-SEG* graphs are contained in IP-SEG graphs. We will later show that this containment is proper. 

We know that other generalizations of permutation and interval graphs like simple triangle or trapezoid graphs remain perfect. One important characteristic that distinguishes IP-SEG and IP-SEG* graphs is that they may contain odd-holes, meaning they are not contained in the class of perfect graphs. IP-SEG* and IP-SEG models of a cycle on five vertices are shown in Figure \ref{fig:ip_seg_c5}. It is easy to see how these models can be extended to models of larger cycles.

\begin{figure}[ht]
\centering
\includegraphics[scale=.55]{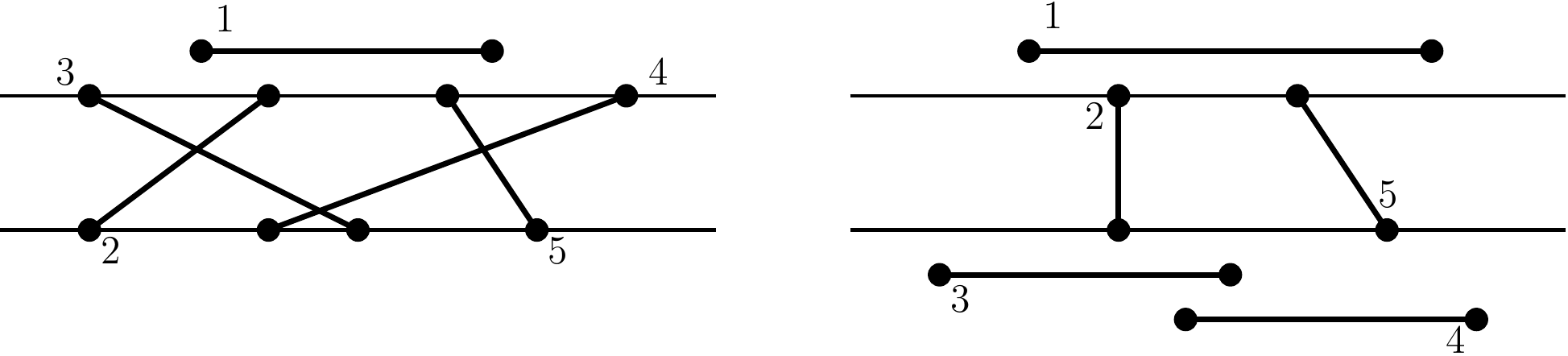}
\caption{An IP-SEG* model of a $C_5$ is shown on the left. An alternative IP-SEG model is given on the right. Note that the numbers shown next to line segments indicate the corresponding vertex labels from a natural labelling of $C_5$.}
\label{fig:ip_seg_c5}
\end{figure}

\section{Representation}
\label{sec:rep_and_count}

Each line segment $s$ in the plane is uniquely determined by its endpoints $s_1$ and $s_2$, each of which can be specified by their $x$ and $y$ coordinates. We will use the notation $s = (x(s_1), y(s_1), x(s_2), y(s_2))$. Thus, when we say we are given the model or representation of an IP-SEG or an IP-SEG* graph $G$, we mean that for each vertex $v$ of $G$ we are given a four-tuple of numbers $(x(v_1), y(v_1), x(v_2), y(v_2))$. In this form, this is not a representation that characterizes IP-SEG or IP-SEG* graphs. In fact it is the characterizing representation of the class of $SEG$ graphs. However, we know that for IP-SEG and IP-SEG* graphs, the $y$ coordinate identifies the horizontal line that an endpoint belongs to, i.e. $y(v_1), y(v_2) \in \{ 1,  2 \}$. $y(v_1) = y(v_2) = y(v)$ indicates that we are dealing with an interval segment and $y(v_1) \neq y(v_2)$ means we have a permutation segment. In addition and without loss of generality, we may assume that for each permutation segment $p$ the first endpoint is positioned on $L_1$ i.e. $p = (x(p_1), 1, x(p_2), 2)$. Further, we may assume that for each interval segment, its left endpoint comes first, i.e. for $i = (x(i_1), y(i), x(i_2), y(i))$, we have $x(i_1) \leq x(i_2)$. We also show that we can limit the range of values that $x(v_1)$ and $x(v_2)$ can take.

Recall that in the model of interval graphs, whether two intervals share a non-empty intersection depends entirely on the relative ordering of their endpoints along the horizontal line. Similarly, whether two permutation segments of a permutation graph representation intersect depends entirely on the relative orderings of their endpoints along $L_1$ and $L_2$. A similar observation can be made for IP-SEG graphs. In particular, for two segments $s$ and $q$, the following holds: 

\begin{enumerate}[i)]
    \item If $s = (x(s_1), y(s), x(s_2), y(s))$ and $p = (x(p_1), y(p), x(p_2), y(p))$ are two interval segments,  they intersect if and only if 
    \begin{itemize}
        \item $y(s) = y(p)$ ($s$ and $p$ lie on the same horizontal line), \textbf{and}
        \item $x(p_1) \leq x(s_1) \leq x(p_2)$ or $x(s_1) \leq x(p_2) \leq x(s_2)$ ($s$ and $p$ overlap).
    \end{itemize}
    \item If $s = (x(s_1), 1, x(s_2), 2)$ and $p = (x(p_1), 1, x(p_2), 2)$ are two permutation segments, they intersect if and only if 
    \begin{itemize}
        \item $x(s_1) \leq x(p_1)$ and $x(p_2) \leq x(s_2)$, \textbf{or}
        \item $x(p_1) \leq x(s_1)$ and $x(s_2) \leq x(p_2)$.
    \end{itemize}
    \item If $s = (x(s_1), y(s), x(s_2), y(s))$ is an interval segment and $p = (x(p_1), 1, x(p_2), 2)$ is a permutation segment, they intersect if and only if $s$ contains one of the endpoints of $p$, that is
    \begin{itemize}
        \item $y(s) = 1$ and $x(s_1) \leq x(p_1) \leq x(s_2)$, \textbf{or}
        \item $y(s) = 2$ and $x(s_1) \leq x(p_2) \leq x(s_2)$.
    \end{itemize}
\end{enumerate}

Thus, all we need to know to determine adjacency between two segments $s$ and $p$ from an IP-SEG or an IP-SEG* model is to know what types of segments $p$ and $s$ are, on which line $L_i$ each of their endpoints is located on, and what are the relative orderings of their endpoints on $L_1$ and/or $L_2$. The last part implies that the number of different values the $x$ endpoint coordinates could take is bounded by the maximum number of endpoints positioned on a single line and, by extension, the total number of endpoints. In other words, we may assume that in an IP-SEG or an IP-SEG* model of a graph $G$, we always have $x(v_1), x(v_2) \in \{ 1,2, \dots, 2n \} $. 

With the above limitations on the range of $x$ and $y$ values, it follows that we need only $O(logn)$ bits of information per vertex to properly represent an IP-SEG or an IP-SEG* graph, meaning these two classes, similar to interval and permutation graphs, have an implicit representation \cite{ImplicitKannan, Efficient}. Therefore, the number of IP-SEG and IP-SEG* graphs on $n$ vertices is bounded by $2^{O(nlogn)}$ 

Note that in the above discussion we have allowed for the possibility of segments sharing one or both endpoints. We know that permutation graphs are usually defined so that no two permutation segments share an endpoint. While this is not imposed as a requirement in the definition of interval graphs, the consecutive cliques arrangement \cite{PQBooth} of an interval graph implies that every interval graph has a representation in which no two intervals share an endpoint. This is also the case for IP-SEG and IP-SEG* graphs.

Indeed, suppose that we have an IP-SEG model of a graph $G$ in which line segments share endpoints and let $e$ be one of those endpoints. Without loss of generality, we may assume $e$ lies on $L_1$. Let $P$ be the set of permutation segments having $e$ as an endpoint. Similarly, let $I_L$ and $I_R$ be the sets of interval segments having $e$ as a left and right endpoint, respectively. We can modify the IP-SEG model of $G$ so that $e$ is no longer a shared endpoint, in the following way. Use the open interval $i_e$ around $e$ that does not contain any other endpoints. In it, arrange the $L_1$ endpoints of permutation segments in $P$ in a reverse order from the one their corresponding $L_2$ endpoints have. In case some of the permutation segments also share an endpoint on $L_2$, resolve the tie arbitrarily. Then, we can extend the interval segments in $I_L$ to the left, so that they all include new endpoints of permutation segments in $P$, end at different endpoints, but their endpoints still remain within $i_e$. We can achieve the same thing with interval segments in $I_R$ by extending them to the right. With this transformation, we ensure that permutations that shared $e$ as an endpoint now either intersect between $L_1$ and $L_2$, or still have a shared endpoint on $L_2$. In addition, all interval segments in $I_L$ and $I_R$ do contain the new endpoints of permutation segments $P$ and all interval segments in $I_L$ and $I_R$ still contain the point $e$, meaning they have non-empty pairwise intersections. Finally, no new pairwise intersections are created as segments in $P \cup I_L \cup I_R$ are the only ones that have an endpoint in the open interval $i_e$. 

Thus, the above transformation leads to a new IP-SEG model of $G$ that has one fewer shared endpoint. Doing this repeatedly, we can obtain an IP-SEG model in which no two segments share an endpoint. The following lemma summarizes the results in this section. 

\begin{lemma} Every IP-SEG graph has an IP-SEG model such that for every vertex $v$, $x(v_1), x(v_2) \in \{ 1,2, \dots , 2n \}$ and no two segments of the model share an endpoint.
\end{lemma}

\section{IP-SEG* and IP-SEG models of chordless cycles}
\label{sec:cycles}

In an IP-SEG model, two interval segments that do not lie on the same parallel line cannot intersect. Let $G$ be a graph with an IP-SEG  model consisting of interval segments only. If $G$ is connected, then all of the interval segments in its IP-SEG model must lie on the same line. By extension, if $G$ has more than one connected component, the interval segments of a single component $C$ must lie on the same line. It is easy to see that the interval segments of a component $C$ on $L_2$ can be translated to $L_1$, while ensuring that they do not intersect with other interval segments already on $L_1$. This means that a graph $G$ has an IP-SEG model consisting of interval segments only, if and only if it has an IP-SEG* model of only interval segments. In other words, $G$ has an IP-SEG model consisting of only interval segments if and only if $G$ is an interval graph. We also have the more straightforward analogous observation for permutation segments: a graph $G$ has an IP-SEG model consisting of only permutation segments if and only if $G$ is a permutation graph.

Interval graphs are chordal and therefore an IP-SEG or an IP-SEG* model of a chordless cycle of length greater than three cannot consist exclusively of interval segments. Similarly, permutation graphs cannot contain an induced cycle on more than four vertices. Therefore, an IP-SEG or an IP-SEG* model of a chordless cycle of length greater than four cannot consist exclusively of permutation segments. A cycle on four vertices, however, may be represented using only permutation segments, as shown in Figure \ref{fig:p_seg_only_c4}.

\begin{figure}[ht]
\centering
\includegraphics[scale=.75]{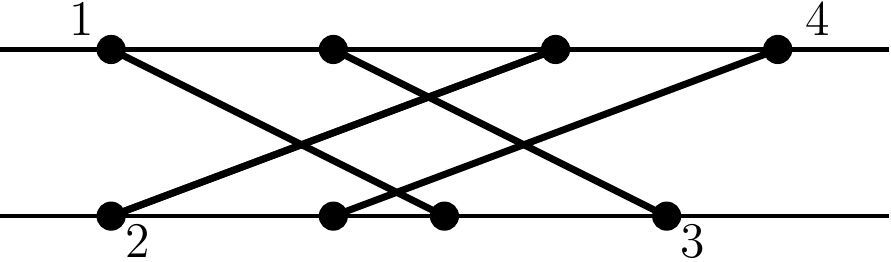}
\caption{An IP-SEG* model of a $C_4$ consisting of permutation segments only}
\label{fig:p_seg_only_c4}
\end{figure}

We have already seen that a $C_5$ is an IP-SEG and an IP-SEG* graph. It is easy to see how the models in Figure \ref{fig:ip_seg_c5} can be extended by adding more interval segments to represent larger chordless cycles. Our goal is to the show that the ways in which such a cycle can be represented is in a certain sense limited and we will use that to identify examples of graphs that do not belong to these two classes.

Let $u$ and $v$ be vertices of an IP-SEG graph $G$ corresponding to interval segments $i_u$ and $i_v$, respectively, in a given IP-SEG representation of $G$. Further, suppose that $i_v$ is fully contained in $i_u$. Then, then $u$ and $v$ are neighbors and all other neighbors of $v$ must also be neighbors of $u$ in $G$. In other words, $N[v] \subseteq N[u]$ must hold for the closed neighborhoods of $v$ and $u$. 

Consider a path $P_n = (v_1,v_2,\dots,v_n)$. Clearly, $N[v_1] \subseteq N[v_2]$, $N[v_{n}] \subseteq N[v_{n-1}]$, and the $v_1$ and $v_n$ are the only vertices in $P_n$ with the property of having a closed neighborhood that is contained within a closed neighborhood of another vertex. Now consider an IP-SEG representation of $P_n$ using only interval segments. From the above discussion it follows that an interval segment could be contained within another only if it corresponds to $v_1$ or $v_n$. We say that an IP-SEG representation of $P_n$ using only interval segments is an \textit{active interval representation} if no interval segment is contained within another. Assuming that the interval segment corresponding to $v_1$ is the one with the leftmost left endpoint, it is easy to see that in an active interval representation of $P_n$, the left-to-right ordering of all left endpoints would coincide with the ordering of the vertices in $P_n$. The same is true for the left-to-right ordering of all right endpoints of interval segments. An example is shown in Figure \ref{fig:act_int_path}. 

\begin{figure}[ht]
\centering
\includegraphics[scale=.75]{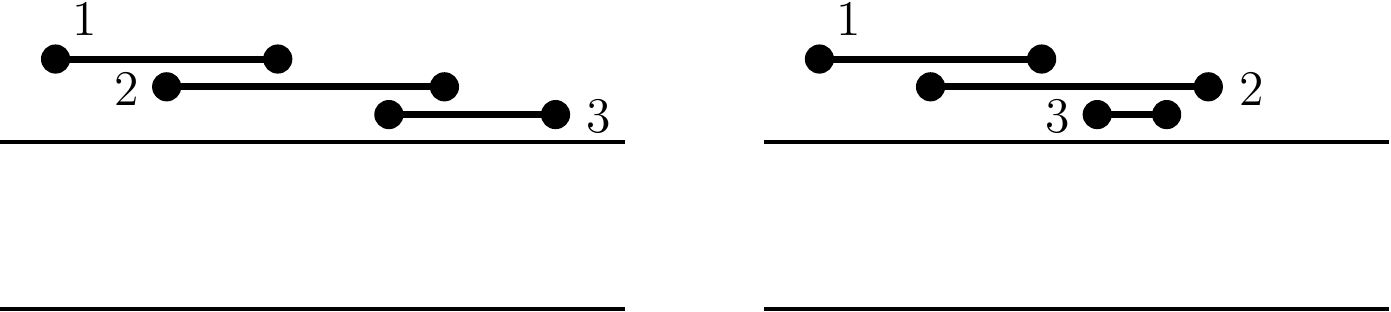}
\caption{An active (left) and an inactive (right) interval representation of a $P_3$.}
\label{fig:act_int_path}
\end{figure}

Suppose we are given an IP-SEG representation IP($C_n$) of a chordless cycle $C_n = (v_1, v_2, \dots, v_n)$, with $n \geq 4$. Let $s(v_i)$ denote the line segment in IP($C_n$) corresponding to vertex $v_i$. Further, suppose that not all $s(v_i)$'s are permutation segments, i.e. that IP($C_n$) consists of both interval and permutation segments. Let $S_{i,j} = (s(v_i), s(v_{i+1}), \dots, s(v_j))$ be a maximal sequence of consecutive interval segments in IP($C_n$), i.e. such that $s(v_{i-1})$ and $s(v_{j+1})$ are permutation segments and for each $k \in \{ i, i+1, \dots, j \}$, $s(v_k)$ is an interval segment (indices are modulo $n$). We call $S_{i,j}$ an \textit{interval arc} of IP($C_n$). Analogously, we can define the notion of a \textit{permutation arc} of IP($C_n$). Using the above notation, any IP-SEG representation IP($C_n$) of a chordless cycle $C_n$ can be described as consisting of interval and permutation arcs. Clearly, the number of interval arcs must equal that of permutation arcs. We also have the following result.

\begin{lemma} 
\label{lem:one_segment_interval_arcs}
Let $C_n$ be a chordless cycle, with $n \geq 4$. $C_n$ has an IP-SEG representation with $t$ interval arcs of lengths $\{ l_1, l_2, \dots, l_t \}$ if and only if $C_{n-n'}$, where $n' = \sum_{i=1}^{t} (l_i - 1)$, has an IP-SEG representation with $t$ interval arcs each of length $1$.  
\end{lemma}

\begin{proof}
Let $S_{i,i+k}$ be an interval arc of a given IP-SEG representation of a chordless cycle $C_n$. Note that $S_{i,i+k}$, being formed by $k+1$ segments, must be an active interval representation of the path $P_{k+1}$. Suppose we take one interval segment $s_j$ in $S_{i,i+k}$ and replace it with two partially overlapping interval segments $s_{j'}$ and $s_{j''}$ such that $s_j = s_{j'} \cup s_{j''}$ and $s_{j'} \cap s_{j''}$ does not overlap any of the other segments in $S_{i,i+k}$. With this we are transforming an IP-SEG representation of $C_{n}$ into an IP-SEG representation of $C_{n+1}$. Similarly, assuming $k \geq 1$, we can take two consecutive interval segments $s_j$ and $s_{j+1}$ in $S_{i,i+k}$ and replace them with one new segment $s_{j*}$ such that $s_{j*} = s_j \cup s_{j+1}$. This allows us to transform an IP-SEG representation of $C_{n}$ with an interval arc of length $k+1$ into an IP-SEG representation of $C_{n-1}$ with an interval arc of length $k$. 
\end{proof}

We are interested in determining how many interval arcs an IP-SEG representation of a chordless cycle can have. Recall that in Figure \ref{fig:ip_seg_c5} we saw one IP-SEG* representation of a $C_5$ consisting of one interval arc and one permutation arc as well as an IP-SEG representation consisting of two interval arcs and two permutation arcs. We will show that these are in fact the only two possible numbers of interval arcs we could have.

\begin{lemma}
\label{lem:IP_SEG*_one_interval_arc}
An IP-SEG* model of a chordless cycle must consist of exactly one interval arc.
\end{lemma}

\begin{figure}[ht]
\centering
\includegraphics[scale=.75]{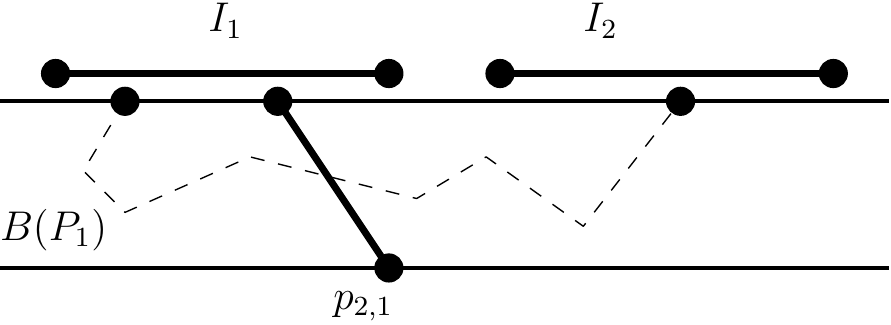}
\caption{An IP-SEG* representation of a chordless cycle cannot consist of two interval arcs because the broken line $B(P_1)$ induced by the first permutation arc $P_1$ will have to intersect a segment of the other permutation arc $P_2$.}
\label{fig:one_interval_arc}
\end{figure}

\begin{proof}
Suppose there exists a chordless cycle $C_n$ with an IP-SEG* model that has two interval arcs.
Then, by Lemma \ref{lem:one_segment_interval_arcs}, there exists a chordless cycle $C_m$, $m \leq n$, with an IP-SEG* model that has two one-segment interval arcs $I_1$ and $I_2$. Without loss of generality, we may assume that in this model $I_1$ and $I_2$ both lie on $L_1$ and $I_1$ is to the left of $I_2$. Let $P_1 = (p_{1,1}, p_{1,2}, \dots , p_{1,q})$ and $P_2 = (p_{2,1}, p_{2,2}, \dots, p_{2,r})$ be the two permutation arcs so that $p_{1,1}$ and $p_{2,1}$ are the permutation segments that intersect $I_1$ and the intersection point $p_{2,1} \cap I_1$ is to the right of $p_{1,1} \cap I_1$. Consider the points $z_0 = p_{1,1} \cap I_1$, $z_{q+1} = p_{q,1} \cap I_2$, and $z_k = p_{1,k} \cap p_{1,k+1}$, $1 \leq k \leq q$. The straight-line segments connecting $z_i$ with $z_{i+1}$, $0 \leq i \leq q$, form a broken line $B(P_1)$ connecting $z_0$ with $z_{q+1}$. $B(P_1)$ is a continuous curve connecting $z_0$ and $z_{q+1}$ that is bounded between $L_1$ and $L_2$. As such, one of its constituting straight-line segments must intersect $p_{2,1}$ (see Figure \ref{fig:one_interval_arc}). Since each of these segments must be contained within some permutation segment of $P_1$, that would imply that there are permutation segments $p_{1,x}$ and $p_{2,y}$ that intersect. However, these two permutation segments correspond to two non-consecutive vertices of a chordless cycle that cannot be adjacent, which is a contradiction. Therefore, a $C_n$ cannot have an IP-SEG* model with two interval arcs. Furthermore, if a $C_n$ could have an IP-SEG* model with $k \geq 3$ interval arcs, we could turn this into an IP-SEG* model with $k-1$ interval arcs by collapsing one of the interval arcs into a single point and effectively merging two permutation arcs into one. Therefore, a $C_n$ cannot have an IP-SEG* model with more than one interval arc.       
\end{proof}

\begin{lemma}
\label{lem:IP_SEG_one_or_two_interval_arcs}
An IP-SEG model of a chordless cycle must consist either of exactly one interval arc or of exactly two interval arcs positioned on different horizontal lines.
\end{lemma}

\begin{figure}[ht]
\centering
\includegraphics[scale=.57]{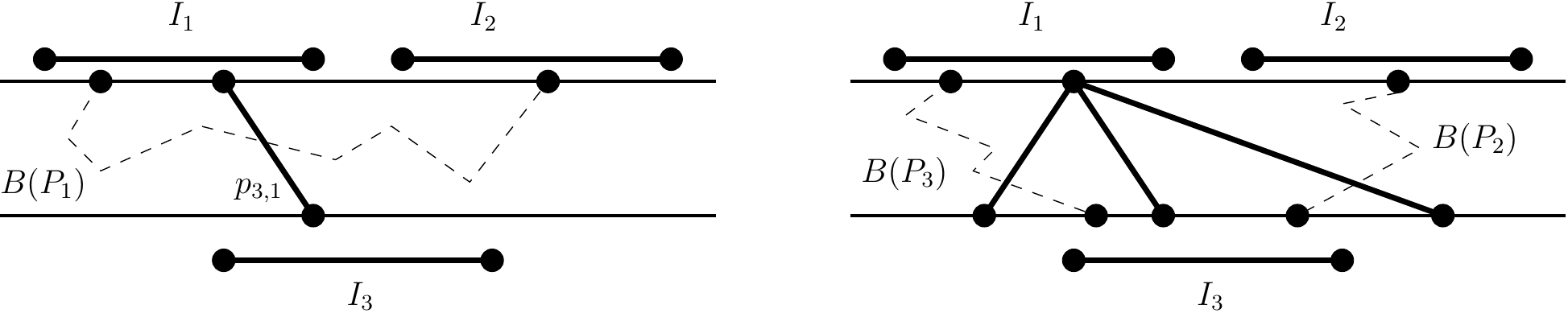}
\caption{An IP-SEG representation of a chordless cycle cannot consist of two interval arcs on $L_1$ and one on $L_2$ because: (1) the broken line $B(P_1)$ induced by the permutation arc $P_1$ will have to intersect a segment of the permutation arc $P_3$ (model on the left) or (2) a segment of the permutation arc $P_1$ would have to intersect either $B(P_2)$, $B(P_3)$, or $I_3$ (model on the right).}
\label{fig:one_interval_arc_per_line}
\end{figure}

\begin{proof}
Given Lemma \ref{lem:IP_SEG*_one_interval_arc}, we only need to consider proper IP-SEG models having at least one interval arc on each horizontal line. Suppose there exists a chordless cycle $C_n$ with an IP-SEG model that has two interval arcs $I_{1}$ and $I_{2}$ on $L_1$ and one interval arc $I_3$ on $L_2$. As in the proof of Lemma \ref{lem:IP_SEG*_one_interval_arc}, we may assume that each of these are one-segment interval arcs.
Let $P_1 = (p_{1,1}, p_{1,2}, \dots , p_{1,q})$ be the arc between $I_1$ and $I_2$, $P_2 = (p_{2,1}, p_{2,2}, \dots , p_{2,r})$ the arc between $I_2$ and $I_3$ and $P_3 = (p_{3,1}, p_{3,2}, \dots , p_{3,s})$ the arc between $I_1$ and $I_3$. Further, assume that $p_{1,1}$ and $p_{3,1}$ are the permutation segments of $P_1$ and $P_3$, respectively, intersecting with $I_1$. If the intersection point $p_{1,1} \cap I_1$ is to the left of $p_{3,1} \cap I_1$, then the broken line $B(P_1)$ induced by the intersection points of permutation segments of $P_1$ (see proof of Lemma \ref{lem:IP_SEG*_one_interval_arc}) must intersect $p_{3,1}$. This would imply an intersection between two permutation segments that correspond to non-consecutive vertices in $C_n$, a contradiction. Suppose the intersection point $p_{3,1} \cap I_1$ is to the left of $p_{1,1} \cap I_1$ and consider the endpoint of $p_{1,1}$ that lies on $L_2$. It cannot lie within $I_3$ as $p_{1,1}$ and $I_3$ do not correspond to consecutive vertices in the chordless cycle. But if it were to lie to the left or right of $I_3$, then $B(P_3)$ or $B(P_2)$ would need to intersect with $p_{1,1}$, both of which lead to a contradiction. Therefore, a $C_n$ cannot have an IP-SEG model with two interval arcs on one horizontal line and one interval arc on the other. Recalling the observation about collapsing arcs, this also implies that an IP-SEG model of a chordless cycle $C_n$ cannot have more than one interval arc per horizontal line.
\end{proof}

We can use Lemma \ref{lem:IP_SEG_one_or_two_interval_arcs} to identify examples of graphs that do not belong to the classes of IP-SEG* and IP-SEG graphs. Consider the graph $G_7$ on 21 vertices formed by a cycle $C_7 = (v_1, v_2, \dots, v_7)$ and vertices $w_i$ and $z_i$, $1 \leq i \leq 7$, such that $w_i$ is pendant to $v_i$ and $z_i$ is pendant to $w_i$. If $G_7$ has an IP-SEG* model, it must consist on one interval arc $I$ and one permutation arc $P$. Let $P = (p_1, \dots, p_k)$, with consecutive segments corresponding to consecutive vertices in $C_7$ and $p_1$ intersecting $I$ to the left of $p_k$. It is easy to see that $2 \leq k$. We also have that $k \leq 6$, since we cannot represent cycles of length more than four with permutation segments only. Let $O_1$ and $O_2$ be the two permutations of $(1, \dots ,k)$ that define the intersections of permutation segments of $P$. $1$ must appear before $k$ in both $O_1$ and $O_2$. Further, since $p_1$ and $p_k$ are the only permutation segments intersecting $I$, $1$ and $k$ must appear consecutively in $O_1$. If $k \geq 5$, then $p_3$ cannot intersect neither $p_1$ nor $p_k$. However, the only way to achieve this is if both $O_1$ and $O_2$ contain the sub-ordering $(1,3,k)$, which contradicts $1$ and $k$ being consecutive in $O_1$. Therefore, $k \leq 4$. 

\begin{figure}[ht]
\centering
\includegraphics[scale=.57]{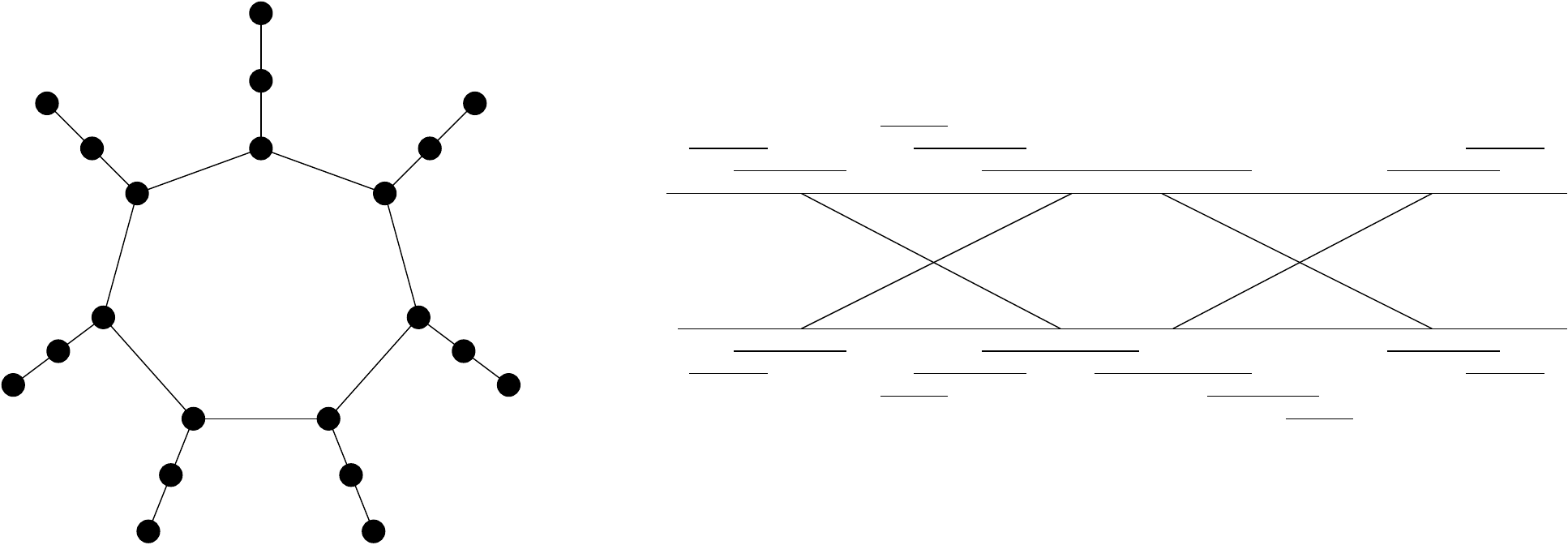}
\caption{The graph $G_7$ (left) is not an IP-SEG* graph, but is an IP-SEG graph. One possible IP-SEG model of $G_7$ is shown on the right.}
\label{fig:g7}
\end{figure}

From the above, it follows that the IP-SEG* representation of $C_7$ must have an interval arc composed of at least three segments. Let $v_t$ be a vertex corresponding to an interior interval segment $s_t$ of $I$. Note that since $s_t$ is between the endpoints of the broken line $B(P)$ induced by the permutation arc, the segment corresponding to $v_t$'s neighbor $w_t$ cannot be a permutation segment. This, combined with the fact that $s_t$'s ends are overlapped by the two neighboring interval segments, $w_t$ must correspond to an interval segment that is fully contained in $s_t$. However, $w_t$ has a neighbor $z_t$ that it does not share with $v_t$, a contradiction. Therefore, $G_7$ is not an IP-SEG* graph. Note that this would be true for any graph $G_n$ constructed in an analogous way from a cycle $C_n$ with $n \geq 7$. $G_7$ is also a minimal non-member of IP-SEG* under vertex removal.  

$G_7$ and $G_8$, can be represented using an IP-SEG model, which demonstrates that the class IP-SEG* is properly contained in IP-SEG. However, this is not true for graphs $G_n$ with $n \geq 9$. One can show that in an IP-SEG model of a chordless cycle that consists of two interval and two permutation arcs, the permutation arcs cannot consist of more than two segments. Thus, at least one interval arc of an IP-SEG representation of a $C_n$ with $n \geq 9$  would need to be of length at least three. This would lead us to the same contradiction when trying to assign segments to $w_t$ and $z_t$ for which $v_t$ corresponds to an interior interval segment.

\section{Clique and independent set}
\label{sec:algorithms}

The respective geometric intersection models of interval and permutation graphs immediately imply very natural polynomial algorithms for solving several important optimization problems, including clique, independent set, and graph coloring. For example, given the model of an interval graph, we can find a maximum independent set by using a greedy algorithm which at each step selects the interval with the leftmost right endpoint, while removing that interval and all other that intersect it from consideration. Given the model of a permutation graph, we can easily find a maximum clique by recovering the defining permutation of the graph from the model and finding the longest decreasing subsequence in it. These algorithms are not of great practical importance, as there exist linear-time algorithms for the respective problems on larger graph classes, that do not require a model as part of the input. Nevertheless, they point us to an initial direction in the study of such problems on the new classes of IP-SEG and IP-SEG* graphs. In particular, we look at the clique and independent set problems on the classes, when the IP-SEG model is given.

\subsection{The Clique problem}

Suppose that $C$ is a clique in an IP-SEG graph $G$ with a given IP-SEG model. Since interval segments lying on different parallel lines cannot intersect, all interval segments in $C$, if any, must lie on a single line. Without loss of generality, we may assume that is $L_1$. Denote by $C_{int}$ the set of interval segments in $C$ and by $C_{prm}$ the set of permutation segments in $C$. 

Consider the case when $C_{int} \neq \emptyset$ and let $s_C$ be the intersection of interval segments in $C_{int}$. As such, $s_C$ is fully contained in each interval segment in $C_{int}$ and it must contain the top endpoint of each permutation segment in $C_{prm}$. In addition, the endpoints of $s_C$ must be endpoints of one or more interval segments in $C_{int}$.

\begin{figure}[ht]
\centering
\includegraphics[scale=.75]{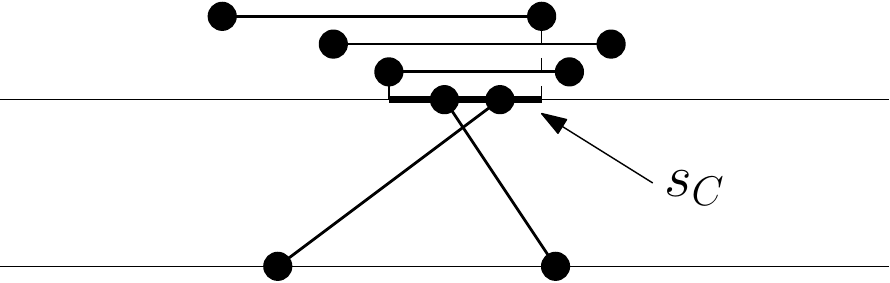}
\caption{An IP-SEG representation of a clique on 5 vertices}
\label{fig:ip_seg_one_clique}
\end{figure}

If we knew what $s_C$ was, it would not be difficult to recover a clique of maximum size. First, set $C_{int}$ to be the set of interval segments fully containing $s_C$. Then, identify all permutation segments that have an endpoint in $s_C$ and find a clique of maximum size $C_{prm}$ on the permutation graph they induce. Note that, since there may be multiple possible options for $C_{prm}$, there may be multiple cliques of maximum size that correspond to $s_C$. Nevertheless, the above procedure will recover a maximum clique, when we know $s_C$.

This leads to a simple polynomial algorithm for the clique problem, when the IP-SEG model is known. First, identify the maximum clique on the graph $G_p$ induced by the permutation segments of $G$. Then, go over all pairs of endpoints $s_1, s_2$ of interval segments along $L_1$ and identify the maximum clique formed by interval segments containing  the interval $s_C = [s_1, s_2]$ and permutation segments with an endpoint in $s_C$. Repeat the same for pairs of endpoints along $L_2$. The clique of largest size found in this procedure is a maximum clique of $G$.

\subsection*{\bf MaxClique (given the model)}
\begin{itemize}
    \item[{\tt 0.}] Find the maximum clique $C_p$ of the graph induced by all permutation segments of $G$
    \item[{\tt 1.}] Set $c_{max} = |C|$ and $C_{max} = C_p $
    \item[{\tt 2.}] For $i$ in $\{ 1,2 \} $
    \item[{\tt 3.}] \quad For each pair of endpoints $s_1, s_2$ on $L_i$ 
    \item[{\tt 4.}] \quad \quad Identify the interval segments on $L_i$ containing the interval $[s_1, s_2]$ and store them in $C_{int}$
    \item[{\tt 5.}] \quad \quad Find a maximum clique formed by permutation segments with an endpoint in $[s_1, s_2]$ and store it in $C_{prm}$  
    \item[{\tt 6.}] \quad \quad Combine $C_{int}$ and $C_{prm}$ into $C$
    \item[{\tt 7.}] \quad \quad \quad If $|C| > c_{max}$ then set $c_{max} = |C|$ and $C_{max} = C$
    \item[{\tt }] Return $C_{max}$
\end{itemize}

Depending on what subroutine we apply to find a maximum clique $C_{prm}$ on the permutation graphs induced by permutation segments, the overall running time of the algorithm would be $O(n^2(n+m))$ or $O(n^3logn)$.

\subsection{The Independent Set problem}

We say that a line segment $p$ (interval or permutation) in an IP-SEG model is \textit{to the right} of another line segment $q$, and write $q < p$, if for each parallel line $L_i$, either at least one of $p$ and $q$ does not have an endpoint on $L_i$ or each endpoint of $p$ on $L_i$ is to the right of each endpoint of $q$ on $L_i$. We say that a line segment $r$ is \textit{between} two segments $p$ and $q$ if $p < r$ and $r < q$. Clearly, two line segments $p$ and $q$ do not intersect if and only if $q < p$ or $p < q$.

Suppose $G$ is an IP-SEG graph with a given model $IP(G)$. As the model is fixed, we will interchangeably use the notions of vertices in $G$ and segments in $IP(G)$, as well as the notions of induced subgraphs of $G$ and sets of line segments in $IP(G)$, depending on the context. Let $I$ be a maximum independent set in $G$ and $I_{prm}$ and $I_{int}$ be the sets of permutation and interval segments, respectively, that form $I$. From the above observation, it follows that the permutation segments in $I_{prm}$ must form a sequence $\{ p_i \}$ such that each $p_{i+1}$ is to the right of $p_i$. 

\begin{figure}[ht]
\centering
\includegraphics[scale=.75]{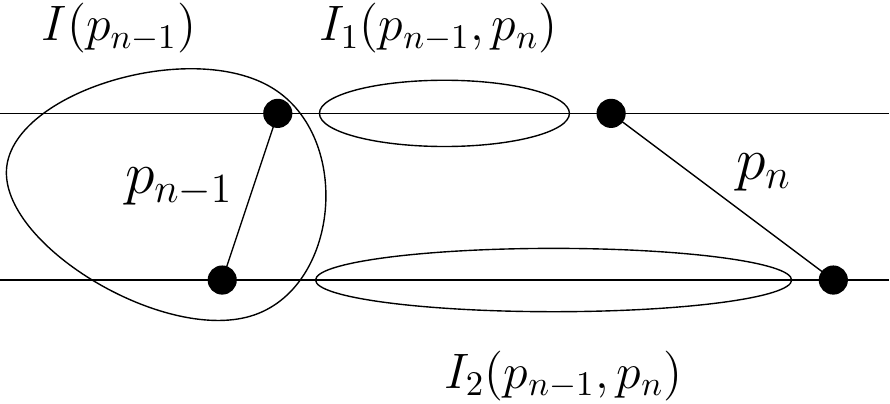}
\caption{An IP-SEG representation of an independent set with at least two permutation segments}
\label{fig:ip_seg_one_ind_set}
\end{figure}

Let $p_{n-1}$ and $p_n$ be the next to last and last permutation segments in the sequence $\{ p_i \}$, respectively. Let $G(p_{n-1})$ be the set of all segments, interval and permutation, that $p_{n-1}$ is to the right of and let $I(p_{n-1}) = I \cap G(p_{n-1})$. It is easy to see that if $I = I(p_n)$ is of maximum size in $G$, then $I(p_{n-1})$ must be an of maximum size in the subgraph $G(p_{n-1})$ of $G$. 

Denote by $G_i(p_{n-1}, p_{n})$ the set of all interval segments along $L_i$ that are between $p_{n-1}$ and $p_{n}$ and let $I_i(p_{n-1}, p_{n}) = I \cap G_i(p_{n-1}, p_{n})$. It is again easy to see that if $I = I(p_n)$ is of maximum size in $G$, then $I_i(p_{n-1}, p_{n})$ must be of maximum size in the subgraph $G_i(p_{n-1}, p_{n})$ of $G$.

While finding $I(p_{n-1})$ amounts to solving the original problem of finding $I = I(p_n)$, given that the subgraph $G_i(p_{n-1}, p_{n})$ is an interval graph, we can find $I_i(p_{n-1}, p_{n})$ by applying an existing algorithm for the independent set on interval graphs.

This leads to a simple dynamic programming algorithm for the independent set problem in which for each permutation segment $p$ we keep track of the largest independent set $I(p)$ formed by $p$ along with segments that $p$ is to the right of. Begin by obtaining a left-to-right topological ordering $T$ of the set of all permutation segments using the "to the right" relation. Then, process segments in $T$ in a left-to-right order. If a segment $p$ is not to the right of any of the already processed segments from $T$, then $I(p)$ is simply the union of $p$ and the largest independent sets on $L_1$ and $L_2$ formed by interval segments that $p$ is to the right of. Otherwise, we also need to consider the sets $I(p') \cup I_2(p',p) \cup I_2(p',p)$ for each permutation segment $p'$ that $p$ is to the right of. We also need to account for the possibility that the largest independent set of $G$ consists only of interval segments in $IP(G)$. For this, we simply need to combine the largest independent set of the interval subgraph of $G$ induced by interval segments lying on $L_1$ with the corresponding independent set on $L_2$.

\subsection*{\bf MaxIS (given the model)}
\begin{itemize}
    \item[{\tt 0.}] Find the largest independent sets $I_1$ and $I_2$ of interval segments on $L_1$ and $L_2$, respectively
    \item[{\tt 1.}] Set $I_{max} = I_1 \cup I_2$ and $i_{max} = |I_{max}|$
    \item[{\tt 2.}] Find a topological ordering $T$ of the permutation segments of $G$
    \item[{\tt 3.}] For $p$ in $T$
    \item[{\tt 4.}] \quad Set $I(p) = I_1(p) \cup I_2(p)$, where $I_i(p)$ is a largest ind. set of interval segments $s$ on $L_i$ such that $s < p$
    \item[{\tt 5.}] \quad For $p'$ in $T$ such that $p' < p$
    \item[{\tt 6.}] \quad \quad $I^*(p) = I(p') \cup I_1(p',p) \cup I_2(p',p)$
    \item[{\tt 7.}] \quad \quad If $|I(p)| < |I^*(p)|$ then set $I(p) = I^*(p)$
    \item[{\tt 8.}] \quad If $|I(p)| > i_{max}$ then set $I_{max} = I(p)$ and $i_{max} = |I(p)|$
    \item[{\tt }] Return $I_{max}$
\end{itemize}

Since we are given the model and thus we have the interval segments in sorted order, we can find each of the independent sets of subgraphs in the algorithm in $O(n)$ time. This leads to an overall running time of $O(n^3)$.

\section{Conclusion and future work}
\label{sec:conclusion}

In this work we introduced two graph classes, IP-SEG* and IP-SEG, generalizing the classes of permutation and interval graphs, based on their geometric models. We showed that these graph classes have an implicit representation. In addition, we saw that unlike earlier generalizations such as simple triangle and trapezoid graphs, these classes are not contained in the class of perfect graphs. Nonetheless, we are somewhat limited in how we can represent a chordless cycle using an IP-SEG model, which leads to some forbidden subgraphs for the two classes. We have also discussed algorithms for the clique and independent set problems on the classes, when an IP-SEG model is given.

The recognition problem is a natural question that remains open. Finding alternative characterizations would be of value in tackling the recognition problem. In particular, given that interval and permutation graphs have nice vertex ordering characterizations and a similar result has been recently obtained for one of their generalizations - simple triangle graphs - by Takaoka \cite{SimpleTriangle}, it would be worth exploring if such a characterization can be found for IP-SEG* or IP-SEG graphs. Another avenue would be to identify other forbidden subgraphs for the two classes. 

Future work should also be done on studying other optimization problems on the class. A good candidate would be coloring, given the simple algorithms for this problem on interval and permutation graphs arising naturally from their geometric intersection models. Finally, it would be interesting to know if we can design robust algorithms for optimization problems such as clique and independent set, for when the IP-SEG model is not given as part of the input.


\bibliographystyle{acm}  
\bibliography{references}  

\end{document}